\documentclass[11pt]{article}

\usepackage{amssymb,amsthm,amsmath}
\usepackage{graphicx}
\pdfoutput=1
\usepackage{geometry}
\usepackage{mathtools}
\usepackage{float}
\usepackage{array}
\usepackage{tikz}
\usepackage{subfig}
\usepackage{authblk}
\usepackage{lipsum}
\usepackage{algorithm, algorithmic}
\usepackage{amssymb}
\newcommand{\R}{{\mathbb R}}
\newfloat{procedure}{htbp}{loa}
\floatname{procedure}{Procedure} %It may be done better
\usepackage{enumitem}
\usepackage[square,numbers]{natbib} %package for manage Bibliography 
\setlist[description]{font=\normalfont\textbullet\space}
\usepackage{titlesec}

\titleformat{\section}
{\normalfont\fontsize{12}{15}\bfseries}{\thesection}{1em}{}
\title{Prior Independent Equilibria and  Linear Multi-dimensional Bayesian Games}

\author[1]{Abbas Edalat}
\author[2]{Samira Hossein Ghorban}
\affil[1]{ Department of Computing, Imperial College London,  United Kingdom\\a.edalat@imperial.ac.uk }
\affil[2]{School of Computer Science, Institute for Research in Fundamental Sciences (IPM), Tehran, Iran\\ s.hosseinghorban@ipm.ir }

\date{}

\newtheorem{theorem}{Theorem}

\newtheorem{corollary}{Corollary}

\newtheorem{definition}{Definition}
\newtheorem{example}{Example}

\newtheorem{lemma}{Lemma}

\newtheorem{proposition}{Proposition}

\begin{document}
\maketitle
	\noindent
	\rule[0.5ex]{1\columnwidth}{1pt}

\begin{abstract}
We show that, in a Bayesian game, a Bayesian strategy map profile is a Bayesian Nash Equilibrium independent of any prior if and only if the Bayesian strategy map profile, evaluated at any type profile, is the Nash equilibrium of the so-called local deterministic game corresponding to that type profile. We call such a Bayesian game, desired in particular in mechanism design when the prior is unknown, type-regular. We then show that an $m$-dimensional $n$-agent Bayesian game whose utilities are linearly dependent on the types of the agents is equivalent, following a normalisation of the type space of each agent into the $(m-1)$-simplex, to a simultaneous competition in $mn$ so-called  basic $n$-agent games. If the game is own-type-linear, i.e., the utility of each agent only depends linearly on its own type, then the Bayesian game is equivalent to a simultaneous competition in $m$ basic $n$-agent games, called a multi-game.  We then prove that an own-type-linear Bayesian game is type-regular if it is type-regular on the vertices of the $(m-1)$-simplex, a result which provides a large class of type-regular Bayesian maps. 

The class of  $m$-dimensional own-type-linear Bayesian games can model, via their equivalence with multi-games, simultaneous decision-making in $m$ different environments. We show that a two dimensional own-type-linear Bayesian game can be used to give a new model of the Prisoner's Dilemma  in which the prosocial tendencies of the agents are considered as their types and the two agents play simultaneously in the PD as well as in a prosocial game. This leads to a type-regular Bayesian game which is proposed as a way of addressing the prosocial tendencies of the agents and the social payoff for cooperation. Similarly, we present a new two dimensional Bayesian model of the Trust game in which the type of the two agents reflect their prosocial tendency or trustfulness, which leads to more reasonable Nash equilibria. We finally consider an example of such multi-environment decision making in production  by several companies in multi-markets.
\end{abstract}

\rule[0.5ex]{1\columnwidth}{1pt}

\vspace{0.5cm}
%\tableofcontents
\newpage
\section{Introduction}
Von Neumann and Morgentern originally modelled the behavior of rational agents in which agents make independent decisions in order to maximize their
utilities, payoffs or self-interests in a single environment or economy~\cite{von2007theory}. 
The notion of Nash equilibrium (NE) has become the key concept in game theory since Nash's celebrated proof of existence of a mixed NE for all finite games~\cite{nash1950equilibrium}.
A similar notion of Bayesian NE is also at the basis of games with incomplete information as shown by Harsanyi~\cite{harsanyi1968games}.

In a game with incomplete information, each agent may have some private information about its strategy set or its payoff which is unknown to the other agents. Harsanyi considered this private information as a set of types for each agent. It is assumed that agents have some belief or prior about other agents' private information, which can be captured by a probability distribution over the types of all agents. Since each agent's choice of strategy depends on its type, one can consider an expanded game in which a Bayesian strategy for an agent is given by a function that takes its types to its choice of actions in the underlying game. Knowing the joint probability distribution of the types, each agent can compute its expected payoff given its own type. This leads to a solution concept of equilibrium that is called pure or mixed Bayesian
NE~\cite[p. 215]{fudenberg1991game}; see Section~\ref{Prior Independent Bayesian Equilibria}.

The first part of this paper is focused on Bayesian games with a Bayesian Nash equilibrium that is independent of any prior. An important area of application of such games is in mechanism design. Auctions and mechanisms are incomplete information games
in which the outcome can be described by the standard equilibrium concept in which a prior assumption is necessary.
To learn the prior, a market analysis can be performed, for example, by hiring
a marketing firm to survey the market and determine distributional estimates
of agent preferences.
This process is quite reasonable in large markets.
In contrast, if there are only a few firms in the market for a product,  the sample size would hardly be enough for estimating the distribution for agents' values~\cite[Chapter 5]{hartlie2017}.
A prior-independent approach to Bayesian game theory means that although there exist prior
distributions from which the agents' values are drawn,
the mechanism designer has no knowledge of these priors
~\cite{devanur2011prior}.

A prior-independent mechanism would be parameterized neither by the distribution on
agent preference nor by the capacity that governs the agents' utility functions.
Based on their results in~\cite{fu2013prior}, the authors argue that it may be possible to develop a general theory for prior-independent mechanisms for risk-averse agents, although this theory would look different from the existing theory of algorithmic mechanism design.

Later in the paper, we examine multi-dimensional Bayesian games with linear utilities. Multi-dimensional Bayesian games, in which the type of each agent is a real vector, have been studied by Krishna and Perry in the context of multiple object auctions~\cite{krishna1998efficient}. 
This class includes all combinatorial auction problems~\cite{lucier2010price}. 
Common examples include  multi-item environments where
an agent has different values for each item~\cite{hartline2013bayesian}: for example, a home buyer may have distinct values for
different houses on the market, an Internet user may have distinct values for various qualities
of service, and an advertiser on an Internet search engine may value traffic for search phrase
“mortgage” higher than that for “loan”~\cite[Chapter 7]{hartlie2017}.
The challenge posed by multi-dimensional private information is that
multi-dimensional type spaces can be large, may be analytically
or computationally intractable~\cite{hartline2013bayesian}.

Athey's result that a monotone pure-strategy equilibrium exists whenever a Bayesian game satisfies the so-called single crossing condition~\cite{athey2001single} has been extended, from one-dimensional type and action spaces, to the setting in which type and action spaces are multidimensional and only partially ordered~\cite{mcadams2003isotone, van2007monotone, reny2011existence}.
In~\cite{rabinovich2013computing}, Bayesian games with multi-dimensional types  where
a utility of an agent depends only on the actions performed by others and not on their type and  each agent draws its type independently from a commonly known continuous distribution have been studied.

The assumption of linear utilities is at the basis of the classical Cournot competition of firms, which has been widely used to model the economy~\cite{daughety2005cournot}. It has been also used recently in linear exchange economies~\cite{bonnisseau2003existence}, for which uniqueness of the equilibrium under the linear assumption and the unique utility level at equilibrium were established in~\cite{bonnisseau2001continuity}. In addition, there is a well-known linear utility representation for a large class of preference relations, i.e., those that are translation invariant and semi-continuous at some point on any finite dimensional Euclidean space~\cite{candeal1995note, trockel1992alternative}.

In this paper, we ask under what conditions a Bayesian game has a prior independent Bayesian Nash equilibrium. We show that a Bayesian strategy map profile, i.e., an assignment of a mixed NE to every type of every agent, is a Bayesian Nash Equilibrium independent of any prior if and only if the Bayesian strategy map, evaluated at any type profile, is the NE of the local game corresponding to that type profile. We accordingly call a Bayesian game with a Bayesian Nash equilibrium independent of any prior a {\em type-regular} Bayesian game. 

We then examine the linear class of multi-dimensional $n$-agent Bayesian games, in which the type space of each agent is a non-zero real vector with, say, $m$ non-negative components, and the utilities of each agent depends linearly on the types of the agents. We show that, following a normalisation of the type space of each agent into the $(m-1)$-dimensional simplex, such a Bayesian game is equivalent to a simultaneous competition by the agents in $mn$ so-called basic $n$-agent games, i.e., one basic game for each type component of the $n$ agents. In case, the utility of each agent only depends linearly on its own type, the Bayesian game is equivalent to simultaneous competition in $m$ basic $n$-agent games, which we call a multi-game. 

 We then prove, using the equivalence with multi-games, that an own-type-linear Bayesian game is type-regular if it is type-regular on the vertices of the $(m-1)$-simplex. This provides a large class of type-regular Bayesian games and thus gives a scheme for designing games with a Bayesian Nash equilibrium independent of any prior. 

Multi-games, as an equivalent representation of own-type-linear multi-dimensional Bayesian games, have some similarities and yet some basic differences with polymatrix games~\cite{yanovskaya1968equilibrium}, one of several well-studied classes of compactly represented games, which also include graphical games~\cite{kearns2001graphical}, hypergraphical game~\cite{papadimitriou2008computing}, and
graphical multi-hypermatrix games~\cite{ortiz2017tractable}. It is useful to recall that a 2-agent Bayesian game with a finite number of types can be represented by a polymatrix game~\cite{howson1974bayesian}.
In a polymatrix game, every agent plays the same strategy in every 2-agent subgame, and its utility is the sum of its subgame utilities.
In a MG, however, the utility of each agent in any local game, i.e, for any type profile, is a weighted sum of its $n$-agent basic game utilities where the weights, considered as private information, are given by the components of the agent's type.

The equivalence with multi-games show us that $m$-dimensional own-type-linear Bayesian games can model simultaneous decision-making by $n$ agents in $m$ different environments, i.e., one environment per dimension. Two different kinds of examples for this form of simultaneous decision-making are considered in this paper: (i) In production  by several companies in multi-markets. (ii) In human decision making in which as well as material payoffs there are social and moral consequences or payoffs for the actions of the agents, as in the Prisoners' Dilemma and the Trust game.

In the economic literature, competition in multi-markets have been examined in the context of the Cournot model. In~\cite{bulow1985multimarket}, Bulow et al.~ provide a numerical example of Cournot markets in which two firms sell in one market and one of them is a monopolist in a second market.
 In more recent years, several authors
%, namely~\cite{bimpikis2014cournot} and \cite{abolhassani2014network},
have examined a network approach to Cournot competition~\cite{bimpikis2014cournot, abolhassani2014network, abolhassani2014network}. In the multi-game approach, different markets are considered as independent of each other with different rates of return and the companies allocate different proportions of their investment, considered as their types, to these markets.

A classic benchmark for modeling human decision making when self-interests are at stake is provided by the Prisoner's Dilemma (PD)~\cite{ostrom2007biography,shubik1970game} in which defection by both agents represents the NE despite the fact that mutual cooperation produces greater reward. This outcome would be consistent with the basic tenet of game theory to maximise self-interests. However, when confronted with the choice to cooperate or defect, human beings not only consider their material score, but also the social and moral implications of their individual decisions, and the consequent social and moral payoffs. This view is supported in game theory by Gintis~\cite[Chapter 1]{gintis2009bounds} and in neuroscience by the finding that decision making has a significant and substantial emotional component~\cite{bechara2000emotion,loewenstein2003role}. In fact, in recent years, there have been experiments on the PD with real people that corroborate this argument empirically.  Khadjavi and Lange present an experiment to compare female prison inmates and students in a simultaneous and an iterated PD~\cite{Khadjavi2013163}. In the simultaneous PD, the cooperation rate among inmates exceeded the rate of cooperating students. The authors have concluded that a similar and significant fraction of inmates and students hold social preferences. Brosig provides findings from a face-to-face experiment that used the PD to analyse whether individuals who possess a willingness to cooperate can credibly signal it and whether it is recognisable by the partner~\cite{Brosig2002275}. Results revealed that both capabilities, signaling and recognising, depend upon the individual's propensity to cooperate. There is a vast literature on non-cooperative games to account for altruistic behaviour; see the literature review in~\cite{chen2011robust} which addresses the issue by considering an altruistic extension of the pay-offs of the agents who are then provided with altruistic coefficients that can be considered as their types. This model however cannot account for the different social values of the agents' different choices of actions.

We show that a two dimensional own-type-linear Bayesian game, equivalent to a double game, i.e., with two basic games, can be used to give a new Bayesian model of the Prisoner's Dilemma  in which the prosocial preferences of the agents are considered as their types. In this model, one dimension or game is represented by the PD and the other dimension is given by a social game which encourages cooperation and allows the social or moral payoffs of the agents' actions to be also taken into account in the decision making.

In more recent years, the so-called Trust game with two agents and an experimenter has been proposed to measure trust in human economic behaviour~\cite{berg1995trust}. Initially the two agents are given an equal amount of money. Then in stage one, the first agent is asked to send some of her money to the experimenter who triples it and sends the tripled amount to the second agent. In stage two, the second agents is asked to send some of the money she has received by the experimenter to the first agent. The NE in the Trust game stipulates that the first agent sends no money to the experimenter and the second agent also sends no money back to the first agent; see Subsection~\ref{trust}. 

However, in practice, as in the PD, human agents deviate from the NE of the Trust game as reported in a meta-analysis of 162 replications of the game involving more than 23,000 participants~\cite{JM11}, in which on average the sender does send some money to the receiver and the receiver does return some of the money received to the sender. An explanation for this deviation has been proposed by evolutionary psychology: “Evolutionary models predict the emergence of trust because it maximises genetic fitness~\cite{berg1995trust}. We use a two dimensional staged Bayesian game or double game in which one game is the Trust game and the other game is a social or conscience game and the prosocial tendencies of the agents are represented by their types. Depending on the receiver's belief, we obtain different subgame perfect Nash equilibria for the double game that include strategy profiles in which both the sender and the receiver forward some money to each other.

\section{Prior independent Bayesian equilibria}\label{Prior Independent Bayesian Equilibria}
We first recall the definition of a general class of Bayesian games as in~\cite[p. 215]{fudenberg1991game}. A Bayesian game $G$ is a game in strategic form with incomplete information which has the following structure: 
	$G=\big<I, (A_i, \Theta_i, u_i )_{i \in I}, p(\cdot)\big>$
	where $I=\{1, \ldots , n \}$ is the set of agents, $A_i$ is agent $i$'s  action set, $\Theta_i$ is agent $i$'s type space, and
	  $u_i: \prod_{i \in I}A_i \times \prod_{i \in I}\Theta_i  \to \mathbb{R}$
	  is  agent $i$'s payoff. 
		The agents' type profile $(\theta_1, \ldots, \theta_n)^t\in \prod_{i\in I} \Theta_i$ is drawn from a given joint probability distribution $p(\theta_1, \ldots, \theta_n)$.
		For any $\theta_i \in \Theta_i$, the function $p(\cdot|\theta_i)$ specifies a
		conditional probability distribution  over $\Theta_{-i}$ representing what agent $i$ 
		believes about the types of the other agents if its own type were $\theta_i$.

The type space of the game is defined as $\Theta:=\prod_{i \in I} \Theta_i$. The pure strategy map space  for agent $i\in I$ is the set
$S^{\Theta_i}_{i}=\{s_i(\cdot):  \Theta_i \to A_{i}\}$ so that $\prod_{i \in I} S_i^{\Theta_i}$ represents the space of all strategy map profiles.
For a strategy map profile $(s_i(\cdot), s_{-i}(\cdot)) \in S_i^{\Theta_i}\times S_{-i}^{\Theta_{-i}}$, the expected utility  of agent $i\in I$ is 
$$
u_i(s_i(\cdot), s_{-i}(\cdot))=	\sum_{\theta_{i} \in \Theta_i}\sum_{\theta_{-i} \in \Theta_{-i}}{p_i(\theta_{-i}|\theta_i) u_{i}\big(s_i(\theta_{i}), s_{-i}(\theta_{-i}),\theta_i, \theta_{-i}\big) }.
$$
Recall that a strategy map profile $(s_1(\cdot), \ldots, s_n(\cdot))$ is a pure Bayesian Nash equilibrium (BNE) if for each agent $i\in I$ and $s'_i(\cdot)\in S_i^{\Theta_i}$,
we have 
$u_i\big(s_i(\cdot), s_{-i}(\cdot) \big) \geq u_i\big(s'_i(\cdot), s_{-i}(\cdot)\big)$ \cite [p. 215]{fudenberg1991game}.
For discrete type spaces, this is equivalent to
$$s_i(\theta_i) \in \arg \max_{a_i \in A_i} \sum_{\theta_{-i}\in \Theta_i}{p_i(\theta_{-i}|\theta_i) u_{i}\big(a_i, s_{-i}(\theta_{-i}),\theta_i, \theta_{-i}\big) },$$
for each $i \in I$ and $\theta_i \in \Theta_i$.
Let $\Delta(A_{i})=\big\{\sigma_i(\cdot):A_{i} \to [0,1]| \sum_{a_i \in A_{i}} \sigma_i (a_i)=1\big\}$ be the set of mixed actions for agent $i\in I$.
By considering the normal form~\cite[p. 3]{fudenberg1991game} of $G$,  the mixed map strategy space for  $G$ is $\big\{\sigma_i(\cdot):\Theta_i \to \Delta(A_i)\big\}$. The notion of mixed Bayesian mixed NE is defined similar to pure BNE.  From now on, we assume a BNE is a mixed Bayesian mixed NE which may be pure.
\begin{definition}
	The restriction of a Bayesian game $G$ to a given type profile  $(\theta_1, \ldots, \theta_n)^t \in \Theta$ is denoted by $G_{(\theta_1, \ldots, \theta_n)}$
	and is called the \emph{local game} for $G$ at $(\theta_1, \ldots, \theta_n)$.	
\end{definition}
In general, any BNE in games with incomplete information
requires the prior distribution to be common knowledge. In many cases, however, the prior distribution may not be known, a situation that for example can occur in mechanism design~\cite{hartlie2017, devanur2011prior, fu2013prior} . In these cases, it is therefore desirable to relax this assumption.
We now seek necessary and sufficient conditions for a Bayesian game to have a
BNE that is independent of any prior.
\begin{theorem}\label{motivation-Regularity}
	Given a Bayesian game $G$, 	
	the strategy map profile $(\sigma_1(\cdot),\ldots, \sigma_n(\cdot))$ is a BNE for all priors if and only if
	the strategy profile  $\big(\sigma_1(\theta_1),\ldots, \sigma_n(\theta_n) \big)$ is a NE for the local game  $G_{(\theta_1,\ldots, \theta_n)}$ for all
	$(\theta_1, \ldots,\theta_n)^t\in \Theta$.
\end{theorem}
\begin{proof}
	We present the proof for the case when all agents have finite type spaces.
	The case of infinite type spaces,  which uses integrals instead of sums to evaluate the payoffs, is entirely similar. 
	First, assume $(\sigma_1(\cdot),\ldots, \sigma_n(\cdot))$ is a BNE for all prior $p$. 
	Hence, for each $i\in I$ and for  any given $\theta_i \in \Theta_i$, we have
	\small
	\begin{align}\label{BNE-inddepensent Prior}
	\sum_{\theta'_{-i} \in \Theta_{-i}} p(\theta'_{-i} | \theta_{i}) u_i(\sigma_i(\theta_i), \sigma_{-i}(\theta'_{-i}),\theta_i, \theta'_{-i})	 \geq 
	\sum_{\theta'_{-i} \in \Theta_{-i}} p(\theta'_{-i} | \theta_{i}) u_i(a_i, \sigma_{-i}(\theta'_{-i}),\theta_i, \theta'_{-i})
	\end{align}
	\normalsize
	for each $a_i \in A_i$ and all priors $p$. 
	For any  given $\theta_{-i} \in \Theta_{-i}$, define the conditional probability distribution: $p(\theta'_{-i}|\theta_i)=1$ if $\theta'_{-i}=\theta_{-i}$ and $0$ otherwise.
	Using this prior $p$ in~(\ref{BNE-inddepensent Prior}), we deduce
	$
	u_i(\sigma_i(\theta_i), \sigma_{-i}(\theta_{-i}),\theta_i, \theta_{-i})
	\geq 
	u_i(a_i, \sigma_{-i}(\theta_{-i}),\theta_i, \theta_{-i})
	$
	for each $a_i \in A_i$. Thus, $(\sigma_i(\theta_i), \sigma_{-i}(\theta_{-i}))$ is a NE for the local game $G_{(\theta_i, \theta_{-i})}$.
	
	Now, suppose the strategy profile  $\big(\sigma_1(\theta_1), \ldots, \sigma_n(\theta_n) \big)$ is a NE for the local game  $G_{(\theta_1, \ldots, \theta_n)}$ for 
	$(\theta_1, \ldots, \theta_n)^t\in  \Theta$. Thus, for each agent $i\in I$:
	\begin{align}\label{NE to BNE}
	u_i(\sigma_i(\theta_i), \sigma_{-i}(\theta_{-i}),\theta_i, \theta_{-i})
	\geq 
	u_i(a_i, \sigma_{-i}(\theta_{-i}),\theta_i, \theta_{-i})
	\end{align}
	for each $a_i \in A_i$.
	Let $p$ be any joint probability distribution on $\Theta$.
	From Inequality~(\ref{NE to BNE}), we obtain:
	$
	p(\theta_{-i} |\theta_i)u_i(\sigma_i(\theta_i), \sigma_{-i}(\theta_{-i}),\theta_i, \theta_{-i})
	\geq 
	p(\theta_{-i} |\theta_i) u_i(a_i, \sigma_{-i}(\theta_{-i}),\theta_i, \theta_{-i})
	$
	for each $a_i \in A_i$.
	Therefore, summing over $\theta_{-i} \in \Theta_{-i}$, we conclude:	
	\begin{align*}
	\sum_{\theta_{-i} \in \Theta_{-i}} p(\theta_{-i} | \theta_{i}) u_i(\sigma_i(\theta_i), \sigma_{-i}(\theta_{-i}),\theta_i, \theta_{-i})	 \geq 
	\sum_{\theta_{-i} \in \Theta_{-i}} p(\theta_{-i} | \theta_{i}) u_i(a_i, \sigma_{-i}(\theta_{-i}),\theta_i, \theta_{-i})
	\end{align*}
	for each $a_i \in A_i$ which shows that
	$(\sigma_i(\cdot),\sigma_{-i}(\cdot))$ is a BNE for the prior $p$.
\end{proof}
Theorem~\ref{motivation-Regularity} motivates the following definition. Let the projection map $\pi_i: \Theta  \to \Theta_i$, for each $i \in I$, be given by $\pi_i(\theta_1, \ldots, \theta_n)=\theta_i$.
\begin{definition}\label{reg-def}
	A  Bayesian game $G$ is \emph{type-regular} on $\Theta'\subseteq \Theta$ if for each agent $i \in I$ there exists a function $\sigma^*_i(\cdot): \pi_i(\Theta') \to \Delta(A_i)$  such that
	the strategy profile  $\big(\sigma^*_1(\theta_1), \ldots, \sigma^*_n(\theta_n) \big)$ is a NE for the local game  $G_{(\theta_1, \ldots, \theta_n)}$ whenever $(\theta_1, \ldots, \theta_n)^t\in \Theta'$.
	If $G$ is type-regular on $\Theta$ then $G$ is simply called \emph{type-regular} and the associated strategy map profile and BNE are also called {\em type-regular}.
\end{definition}

Intuitively,  a Bayesian game $G$ is type-regular on
$\Theta'$ if for each agent and a given type component for it, selected from the set $\Theta'$,
the agent can select an action, dependent only on the given type component, such that for each type profile in $\Theta'$ the resulting action profile is a NE for the local game specified by that type profile.
Note from the definition that if $G$ is type-regular on $\Theta'\subseteq \Theta$, then it is type-regular on any subset of $\Theta'$. Theorem~\ref{motivation-Regularity} can be thus reformulated: a Bayesian game has a prior independent BNE  iff it is type-regular.

\section{Linear multi-dimensional Bayesian games}\label{vector Typed Bayesian Games}

We consider the standard Cartesian coordinate system in $\R^m$ for a given integer $m>1$ with the standard basis vectors $v_j=(v_{j1},\ldots,v_{jm})^t\in \R^m$, where $v_{jr}=1$ for $j=r$ and $0$ otherwise. Let $V=\{v_j:1\leq j\leq m\}$ and $\R_+$ denote the set of nonnegative real numbers. 
\begin{definition}
	A Bayesian game $G$ is {\em multidimensional} if the type of each agent is a   vector in ${\mathbb R}_+^m \setminus\{0\}$. A multidimensional Bayesian game is {\em type-linear } if the utility of each agent depends linearly on the types of all agents.  A linear Bayesian game is {\em own-type-linear}
	if the utility of each agent only depends on its own type. 
\end{definition}

We will now show that an $m$-dimensional $n$-agent Bayesian game is equivalent to a simultaneous competition by the $n$ agents in $m$ basic $n$-agent games.
Let $\Sigma^{m-1}=\{x\in \R_+^m:\sum_{i=1}^mx_i=1\}$ be the $(m-1)$-dimensional simplex in $\R^m$.  The following notion of a multi-game is similar to that introduced in~\cite{edalat2012multi}. 
\begin{definition}\label{Def:MG}
	A {\em  multi-game} 
		$G=\big<I,J, \{G_{j}\}_{j \in J},  \{A_{i}\}_{i\in I}, \{\Theta_i\}_{i \in I}, \{u_{ij}\}_{i\in I,  j\in J}, p(\cdot)\big>$ is an  own-type-linear Bayesian game  given by		
		$\big<I, (A_i, \Theta_i, u_i )_{i \in I}, p(\cdot)\big>$ with the following conditions:
		\begin{enumerate}
			\item[(1)] Agent $i$'s type space $\Theta_i\subseteq \Sigma^{m-1}$ for $i\in I=\{1,\ldots,n\}$.
			\item[(2)] There is a set of $n$-agent basic games $G_{j}$ where $j \in J=\{1,\ldots, m\}$ with action space $A_{i}$ and payoff function $u_{ij}$ for agent $i \in I$.
			\item [(3)] Agent $i$'s payoff for the strategy profile $(s_i, s_{-i})$ and type profile $(\theta_i, \theta_{-i})$ is given by
			\begin{equation*}
			u_{i}(s_i, s_{-i}, \theta_i, \theta_{-i})= \sum_{j\in J} u_{ij} (s_i,s_{-i})\theta_{ij}.	
			\end{equation*}			
		\end{enumerate}
\end{definition}
Multi-games (MG) can be seen to have some similarities and yet some basic differences with polymatrix games as pointed out in the Introduction. We now generalise the notion of a multi-game so that the utility of each agent depends on the types of all agents. A  {\em  generalized multi-game} 
		$$G=\Big<I,J, \{G_{kj}\}_{k\in I,j \in J},  \{A_{i}\}_{i\in I}, \{\Theta_i\}_{i \in I}, \{u_{ikj}\}_{i,k\in I,  j\in J}, p(\cdot)\Big>$$
		is a type-linear Bayesian game given by
		$\big<I, (A_i, \Theta_i, u_i )_{i \in I}, p(\cdot)\big>$ 
		which satisfies  items (1) and (2) of Definition~\ref{Def:MG} with item (3)  replaced by the following item:\\[.2ex] 
		(3') Agent $i$'s payoff for the strategy profile $(s_i, s_{-i})$ and type profile $(\theta_i, \theta_{-i})$ is given by
		$$
		u_{i}(s_i, s_{-i}, \theta_i, \theta_{-i})= \sum_{k\in I,j\in J} u_{ikj} (s_i,s_{-i})\theta_{kj}.	
		$$
We can now show the following result.
\begin{theorem}\label{General MG}
	Suppose $G$ is a type-linear $n$-agent Bayesian game such that for each agent $i\in I$ the type space $\Theta_i \subseteq \mathbb{R_+}^{m}\setminus\{0\}$, then $G$ is equivalent with a generalized MG. 
\end{theorem}
\begin{proof}
	Let 	$G=\big<I, (A_i, \Theta_i, u_i )_{i \in I}, p(\cdot)\big>$. Since $\theta_i\neq 0$ for $\theta_i\in \Theta_i$ for each agent $i\in I$, we can divide each vector $\theta_i\in \Theta_i$ by $\sum_{j\in J}\theta_{ij}>0$ and assume that $\Theta_i\subseteq \Sigma^{m-1}$ for $i\in I$. 
	Since $G$ is type-linear, for $i\in I$ and $(s_i,s_{-i})\in A_i\times A_{-i}$, there exists $L_{ik}(s_i,s_{-i}) \in \mathbb{R}^{m}$, for $k\in I$, such that
	$
	u_i(s_i,s_{-i},\theta_1,\ldots, \theta_n)=\sum_{k\in I}(L_{ik}(s_i,s_{-i}))^t\theta_k.
	$
	Consider the generalised MG given by \\[.2ex]
	\begin{center}$\hat{G}=\Big<I, J, \{G_{kj}\}_{k\in I,j \in J}, \{\Theta_i\}_{i \in I}, \{A_{i}\}_{i\in I},  \{\hat{u}_{ikj}\}_{i,k\in I, j\in J}, p(\cdot)\Big>$\\[.2ex]\end{center}
	 such that agent $i$'s utility function for the
	basic game $G_{kj}$ is given by  
	$\hat{u}_{ikj}(s_i,s_{-i})=(L_{ik}(s_i,s_{-i}))_j$ for $i\in I$ and $j\in J$.
	 Agent $i$'s utility in $\hat{G}$ is now seen to be that in $G$ as follows:
	$$\hat{u}_i(s_i,s_{-i}, \theta_i, \theta_{-i})=\sum_{k\in I}\sum_{j\in J} \hat{u}_{ikj}(s_i,s_{-i})\theta_{k{j}}=\sum_{k\in I}\sum_{j\in J}(L_{ik}(s_i,s_{-i}))_j\theta_{k{j}} =u_i(s_i,s_{-i},\theta_1,\ldots, \theta_n) .$$
\end{proof}

\begin{corollary}\label{MG-equivalent} 
	Suppose $G$ is an own-type-linear $n$-agent Bayesian game such that the type space $\Theta_i \subseteq \mathbb{R_+}^{m}\setminus\{0\}$ for each $i\in I$, then $G$ is equivalent with a  MG. 
\end{corollary}
\begin{proof}
		Let 
	$\hat{G}=\big<I, J, \{G_{j}\}_{j \in J}, \{\Theta_i\}_{i \in I}, \{A_{i}\}_{i\in I},  \{\hat{u}_{ij}\}_{i\in I, j\in J}, p(\cdot)\big>$
	be a  MG such that agent $i$'s utility function for the
	basic game $G_{j}$ is given by $\hat{u}_{ij}=\hat{u}_{iij}$ for each $i \in I$, $j\in J$, where $\hat{u}_{iij}$ is as in the proof of Theorem~\ref{General MG} (with $\hat{u}_{ikj}=0$ for $k\neq i$).
\end{proof}	Note that the normalisation in the proof of Theorem~\ref{General MG} implies that any type $\theta_i\in \Theta_i$ satisfies $\sum_{j=1}^m\theta_{ij}=1$. For a eneralised multi-game, this means that the type component $\theta_{ij}\geq 0$ is the proportion of the agent $i$'s type allocated to the basic game $G_{ij}$. For a multi-game, this means that the type component $\theta_{ij}\geq 0$ is the proportion of the agent $i$'s type allocated to the basic game $G_j$. 

Multi-games, equivalently own-type-linear Bayesian games, can therefore model the behaviour
of a finite number of rational agents who play in a number of different environments
simultaneously, where each environment is represented by a basic game and
the resources of each agent are allocated with varying proportions, as private
information, to these basic games. 
\section{Regular Bayesian Nash equilibrium}
By Theorem~\ref{motivation-Regularity},  a Bayesian game has a BNE for all prior if and only if it is type-regular.
In this section,  we investigate a necessary and sufficient condition for type-regularity of a type-linear multi-dimensional Bayesian game $G$. 
 By normalization, we can assume that $\Theta_i \subset \Sigma^{m-1}$. Since utility functions are linear on  types, we can extended these utility functions to $\Sigma^{m-1}$. Therefore, without loss of generality, from now on, we assume $\Theta_i=\Sigma^{m-1}$. Thus, the type space is $\Theta=\big(\Sigma^{m-1}\big)^n$ and its boundary is given by $\bigcup_{i\in I}\Theta_i\times V^{n-1}$, where for clarity we have written agent $i$'s type space as $\Theta_i$ rather than $\Sigma^{m-1}$ and $V$, recall, is the set of vertices of $\Sigma^{m-1}$. Thus, $V^n$ is the set of vertices of $\Theta$.
We aim to show that an own-type-linear Bayesian game $G$ is type-regular if and only if $G$ is type-regular on $V^n$. The proofs of the following two lemmas are given in Section~\ref{appendix}.
\begin{lemma}\label{Simlex-Make}
	If an own-type-linear Bayesian game $G$ is  type-regular on $V \times \{\theta_{-i}\}$ for a given $i\in I$ and $\theta_{-i} \in \Theta_{-i}$, then it is type-regular on  $\Theta_i \times \{\theta_{-i}\}$.
\end{lemma}

\begin{corollary}\label{extention-regularity}
	Suppose an own-type-linear Bayesian game  is type-regular on $V \times \{\theta_{-i}\}$ for a given $i\in I$ and $\theta_{-i} \in \Theta_{-i}$ with  $\sigma^*_i(\cdot): V \to \Delta(A_i)$  and $\sigma_{-i} \in \Delta(A_{-i})$ as a witness for  type-regularity.
	Then the extended function $\overline{\sigma^*_i}(\cdot): \Theta_{i} \to \Delta(A_i)$ with
	$\overline{\sigma^*_i}(\theta_{i})=	\sum_{j=1}^{m}\theta_{ij}\sigma^*_i(v_j)$
	induces type-regularity on $\Theta_i\times \{\theta_{-i}\}$ .
\end{corollary}
Recall that $V^n$ is the set of vertices of the type space $\Theta$.

\begin{lemma}\label{Simlex-Make-Booundary}
	If an own-type-linear Bayesian game  $G$ is type-regular on $V^n$, then $G$ is type-regular on the boundary $\bigcup_{i\in I} \Theta_i\times V^{n-1}$ of the type space $\Theta$.
\end{lemma}

\begin{theorem}\label{Bounadary}
	If an own-type-linear Bayesian game  is type-regular on the boundary $\bigcup_{i\in I} \Theta_i\times V^{n-1}$ of the type space $\Theta$, then it is type-regular.
\end{theorem}
\begin{proof}
   By Theorem~\ref{MG-equivalent}, we can assume  $G$ is a type-regular multi-game on $\bigcup_{i\in I} \Theta_i\times V^{n-1}$. Then for each $i \in I$, there exists
   $\sigma^*_i: \Theta_i \to \Delta(A_i)$ and 
   $\sigma^*_{-i}: V^{n-1} \to \Delta(A_{-i})$ 
		such that 
		$(\sigma^*_i(\theta_i), \sigma^*_{-i}(\theta_{-i}) )$ is a NE for $G_{(\theta_i, \theta_{-i})}$  for each
		$(\theta_{i}, \theta_{-i}) \in \Theta_{i} \times V^{n-1}$.
		We claim that $(\sigma^*_1(\theta_{1}), \ldots,\sigma^*_n(\theta_{n}))$  is a NE for 
		$G_{(\theta_1, \ldots, \theta_n)}$ for each $(\theta_1, \ldots, \theta_n)^t \in \Theta$.
		For each agent $i\in I$, we have
		\begin{align*}
		u_i\big(\sigma^*_i(\theta_i), \sigma^*_{-i}(\theta_{-i}),\theta_i,\theta_{-i}\big)&= 
		\sum_{j=1}^{m}{\theta_{ij}u_i\big(\sigma^*_i(\theta_i), \sigma^*_{-i}(\theta_{-i}), v_j, \theta_{-i}\big)}
		\end{align*}
		Since  $(\sigma^*_i(\theta_i), \sigma^*_{-i}(\theta_{-i}))$ is a NE for 
		 the local game $G_{(v_{j_1}, \ldots, v_{j_{i-1}},\theta_{i}, v_{j_{i+1}}, \ldots, v_{j_n})}$ we have
		$$\sum_{j=1}^{m}{\theta_{ij}u_i\big(\sigma^*_i(\theta_i^j), \sigma^*_{-i}(\theta_{-i}), \theta_i, \theta_{-i}\big)}
		\geq 
		\sum_{j=1}^{m}{\theta_{ij}u_i\big(\sigma'_i, \sigma^*_{-i}(\theta_{-i}), v_j, \theta_{-i}\big)},
		$$		for each $\sigma'_i \in \Delta(A_i)$.	
		Hence 
		$u_i\big(\sigma^*_i(\theta_i), \sigma^*_{-i}(\theta_{-i}) ,\theta_i,\theta_{-i}\big) \geq u_i\big(\sigma'_i, \sigma^*_{-i}(\theta_{-i}),\theta_i,\theta_{-i}\big)$
		which implies that $G$ is type-regular.
\end{proof}
From Lemma~\ref{Simlex-Make-Booundary} and Theorem~\ref{Bounadary}, we obtain:
\begin{corollary} \label{suffucient and nassacery condition for multi-game}
	If an own-type-linear Bayesian game  $G$  is type-regular on $V^n$, then it is type-regular and thus has a prior independent BNE.
\end{corollary}
Note that in Corollary~\ref{suffucient and nassacery condition for multi-game}, the agents' type spaces can be infinite or continuous and we will still have a prior independent BNE. Corollary~\ref{suffucient and nassacery condition for multi-game} provides a method to construct type-regular Bayesian games by ensuring that it is type-regular on $V^n$ as we will see in the next two sections.

\section{Examples of decision making in multi-environments}\label{vector Typed Bayesian Games}

In Section~\ref{appendix}, we indicate how own-type-linear Bayesian games, equivalently, multi-games, can be used to model production by several companies in multi-markets. In this section, we show how multi-games can be used to create new two dimensional Bayesian models for well-known the PD and the Trust game. 

A multi-game is called a \emph{Double Game} (DG)  if $m=2$. In a DG, it is convenient to write the type $(\theta_{i1},\theta_{i2})$ of agent $i$ as $\theta_i:=\theta_{i2}$ with $\theta_{i1}=1-\theta_{i2}=1-\theta_i$. Thus, for an $n$-agent DG we have $\Theta_i \subseteq [0,1]$ and $\Theta=\prod_{i \in I} \Theta_i =[0,1]^n$ with the boundary of $\Theta$ given by $\{0,1\}^n$.

\subsection{A Double Game for Prisoner's Dilemma }~\label{dgpd} As argued in the Introduction, in many circumstances, human beings consider not only their material score, but also the social payoffs of any decision they
make. 
This can be modelled by allowing agents to engage simultaneously in a social game and a standard material game for their utilities. We will show this for the case of the PD.

Consider the standard PD with the payoffs as given in Table~\ref{tab:6} (left) with $t > r> p > s$ and $r>(t+s)/2$ as in~\cite{axelrod2006evolution}.
The social game (SG) encourages cooperation and discourages defection,
as cooperating is usually considered to be the right ethical and moral choice  when interacting with others in social dilemmas. This can be done in different ways which correspond to different types of payoff matrices. Here, we will only consider the case in which  SG encourages cooperation and discourages defection for each agent, independently of the action chosen by the other agent. We present the normal form and the mathematical formulation of the SG as follows. Assume that the competing participants in the SG are agents 1 and 2. Each of them can select  $C$ or $D$.
When they have both made their choice, the payoffs assigned to them are calculated according to 
Table~\ref{tab:6}(right).
\begin{table} 
	\centering
	\begin{tabular}{|c|c|c|}
		\cline{2-3} % scpecifies which column should have border in this row
		\multicolumn{1}{c|}{} & $C$ & $D$  \\ \hline
		$C$ &$(r, r) $&$(s, t)$ \\\hline
		$D$ &$(t, s) $&$(p, p)$\\\hline
	\end{tabular}\qquad
	\begin{tabular}{|c|c|c|}
		\cline{2-3} % scpecifies which column should have border in this row
		\multicolumn{1}{c|}{} & $C$ & $D$  \\ \hline
		$C$ &$(y, y) $&$(y, z)$ \\\hline
		$D$ &$(z, y) $&$(z, z)$\\\hline
	\end{tabular}
	\caption{Payoff metrics for PD (left) and SG (right)}
	\label{tab:6}
\end{table}
As in~\cite{ounsley2010dissertation}, let $y>z=s$, i.e., the SG encourages cooperation and the least payoff of the PD and SG are taken to be the same. The strategy profiles $(D,D)$ and $(C,C)$ are NEs for PD and SG, respectively. Consider a DG with basic games PD and SG.
The type $\theta_i\in[0,1]$ of agent $i=1,2$ is their {\em prosocial} coefficient, with $\theta_i=0$ reflecting  complete selfishness while $\theta_i=1$ indicating  maximum pro-sociability. It is easy to see that the DG is type-regular on $\{0,1\}^2$. By Corollary~\ref{suffucient and nassacery condition for multi-game}, it is type-regular with type-regularity witness
$\sigma^*_i(\cdot): \Theta_{i} \to \Delta(\{D, C\})$ with
$\sigma^*_i(\theta_{i})=	(1-\theta_i)D + \theta_i C$ for each agent $i=1,2$.
By Theorem~\ref{motivation-Regularity}, $(\sigma^*_1(\cdot), \sigma^*_2(\cdot))$ is a BNE for all priors.

This framework for considering the PD with a SG, we propose, is a way to model real-life situations, as, in general, decisions based on prosocial or moral incentives and beliefs do not necessarily bring the highest material benefits. 

\subsection{A Double Game for Trust Game}\label{trust}
Trust game is a 2-agent stage game $G_1$ in which 
$A_1=[0,1]$, $A_2=\{x|3y\geq x, y \in A_1\}$
and $ u_{1}(y,x)=x-y$, $u_{2}(y,x)=3y-x$
for $y\in A_1$ and $x\in A_2$. 	By  backward induction, when the first agent plays first, $(0,0)$ is the NE.
If, for the sake of illustration, we restrict agent 1's actions to $A_1'=\{0,1\}$, then  Figure~\ref{Fig:DG-for-TrustGame} shows the branches of the stage game where the two agents are named $ag_1$ and $ag_2$, respectively. As usual, the label on each edge is the action taken by the agent on the node above and udder each leaf, the first number is the payoff of agent 1 for the branch corresponding to the leaf and the second number is agent 2's payoff.

Under the standard economic assumption of rational self-interest, the predicted actions of the first agent in Trust game will be to send nothing, and any behaviour that deviates from this self-interest is viewed as irrational. Since in actual experiments, individuals significantly deviate from this NE, we argue that, as well as their material interest, they seek to build or protect their social reputation or their own ethical and prosocial values. We thus propose to develop a more realistic model of trust in economic behaviour by using a DG which includes Trust game above and a second social or conscience game, which is very basically formulated here, a follows. 

Let $G_1$ be Trust game as described above and let $G_2$ be the associated {\em conscience or social} game in which 
$A_1=[0,1]$, $A_2=\{x\geq 0|3y\geq x, y \in A_1\}$
and 
$ u_{12}(y,x)=y$ and $ u_{22}(y,x)=x-2y $
for $y\in A_1$ and $x\in A_2$.	The social payoff of agent 1 in sending amount $y$ to agent 2 is considered as the amount $y$ itself. If agent 2 receives $3y$ as a result of agent 1 sending amount $y$, it is only fair that the increment namely $2y$ be divided equally between the two agents. This stipulates the base line that agent 2 sends amount $2y$ to agent 1 so that each would have gained amount $y$. Thus, the social payoff to agent 2 is taken as the linear function $x-2y$, which is positive or negative depending on whether more than or less than $2y$ is sent back by agent 2. By  backward induction, $(1,3)$ is the NE for $G_2$.
If again, we restrict agent 1's actions to $A'_1=\{0,1\}$, then  Figure~\ref{Fig:DG-for-TrustGame} shows the branches of the stage game. 
\begin{figure}
	\centering
	\subfloat{	
		\begin{tikzpicture}[scale=.5]%[baseline=(v1.base)]
		\draw[black, thick] (0,5) -- (3,3);
		\draw[black, thick] (0,5) -- (-3,3);
		\node at (-2,4.3) {\color{blue}{$0$}};
		\node at (2,4.3) {\color{blue}{$1$}};
		\draw[black, thick] (-3,3) -- (-1.5-3,0);
		\draw[black, thick] (3,3) -- (1,0);
		\draw[black, thick] (3,3) -- (1+1.5,0);
		\draw[black, thick] (3,3) -- (1+3,0);
		\draw[black, thick] (3,3) -- (1+4.5,0);
		\node at (-1.5-3-.2,.75) {\color{red}{$0$}};
		\node at (1-.1,.75) {\color{red}{$0$}};
		\node at (1+1.5-.3,.75) {\color{red}{$1$}};
		\node at (1+3+.3,.75) {\color{red}{$2$}};
		\node at (1+4.5+.2,.75) {\color{red}{$3$}};
		\fill[blue!50!black] (0,5) circle (.1);
		\node at (0.3,5.3) {\color{blue}{$ag_1$}};
		\fill[red!50!black] (3,3) circle (.1);
		\node at (3+.2,3.4) {\color{red}{$ag_2$}};
		\fill[red!50!black] (-3,3) circle (.1);
		\node at (-3-.3,3.4) {\color{red}{$ag_2$}};
		\node at (-1.5-3,0-1.3) {\color{red}{$0$}};
		\node at (-1.5-3,0-.5) {\color{blue}{$0$}};
		\node at (1-.2,0-.5) {\color{blue}{$-1$}};
		\node at (1,0-1.3) {\color{red}{$3$}};
		\node at (1+1.5,0-.5) {\color{blue}{$0$}};
		\node at (1+1.5,0-1.3) {\color{red}{$2$}};
		\node at (1+3,0-.5) {\color{blue}{$1$}};
		\node at (1+3,0-1.3) {\color{red}{$1$}};
		\node at (1+4.5,0-.5) {\color{blue}{$2$}};
		\node at (1+4.5,0-1.3) {\color{red}{$0$}};
		\fill[black!50!black] (-1.5-3,0) circle (.1);
		\fill[black!50!black] (1,0) circle (.1);
		\fill[black!50!black] (1+1.5,0) circle (.1);
		\fill[black!50!black] (1+3,0) circle (.1);
		\fill[black!50!black] (1+4.5,0) circle (.1);
		\node at (0.3,-2.5) {$G_1$};
			\end{tikzpicture}
		}
	\hspace{2cm}
	\centering
	\subfloat{	\begin{tikzpicture}[scale=.5]%[baseline=(v1.base)]
		\draw[black, thick] (0,5) -- (3,3);
		\draw[black, thick] (0,5) -- (-3,3);
		\node at (-2,4.3) {\color{blue}{$0$}};
		\node at (2,4.3) {\color{blue}{$1$}};
		\draw[black, thick] (-3,3) -- (-1.5-3,0);
		\draw[black, thick] (3,3) -- (1,0);
		\draw[black, thick] (3,3) -- (1+1.5,0);
		\draw[black, thick] (3,3) -- (1+3,0);
		\draw[black, thick] (3,3) -- (1+4.5,0);
		\node at (-1.5-3-.2,.75) {\color{red}{$0$}};
		\node at (1-.1,.75) {\color{red}{$0$}};
		\node at (1+1.5-.3,.75) {\color{red}{$1$}};
		\node at (1+3+.3,.75) {\color{red}{$2$}};
		\node at (1+4.5+.2,.75) {\color{red}{$3$}};
		\fill[blue!50!black] (0,5) circle (.1);
		\node at (0.3,5.3) {\color{blue}{$ag_1$}};
		\fill[red!50!black] (3,3) circle (.1);
		\node at (3+.2,3.4) {\color{red}{$ag_2$}};
		\fill[red!50!black] (-3,3) circle (.1);
		\node at (-3-.3,3.4) {\color{red}{$ag_2$}};
		\node at (-1.5-3,0-1.3) {\color{red}{$0$}};
		\node at (-1.5-3,0-.5) {\color{blue}{$0$}};
		\node at (1,0-.5) {\color{blue}{$1$}};
		\node at (1-.2,0-1.3) {\color{red}{$-2$}};
		\node at (1+1.5,0-.5) {\color{blue}{$1$}};
		\node at (1+1.5,0-1.3) {\color{red}{$-1$}};
		\node at (1+3,0-.5) {\color{blue}{$1$}};
		\node at (1+3,0-1.3) {\color{red}{$0$}};
		\node at (1+4.5,0-.5) {\color{blue}{$1$}};
		\node at (1+4.5,0-1.3) {\color{red}{$1$}};
		\fill[black!50!black] (-1.5-3,0) circle (.1);
		\fill[black!50!black] (1,0) circle (.1);
		\fill[black!50!black] (1+1.5,0) circle (.1);
		\fill[black!50!black] (1+3,0) circle (.1);
		\fill[black!50!black] (1+4.5,0) circle (.1);
			\node at (0.4,-2.5) {$G_2$};
		\end{tikzpicture}
		}
	\caption{Trust Game $(G_1)$, Social Game $(G_2)$. }
	\label{Fig:DG-for-TrustGame}
\end{figure}
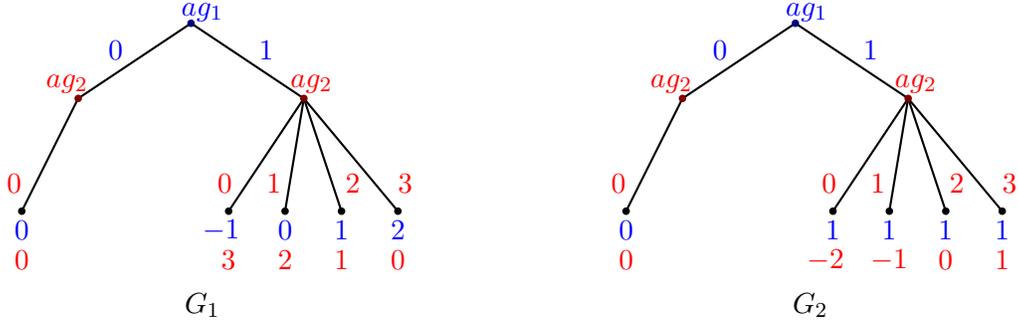

Consider a double game $G$ with basic games  $G_1$ and $G_2$ where
$A_1=[0,1]$, $A_2=\{x\geq 0|3y\geq x, y \in A_1\}$,  $\Theta_1=\{1/4\}$ and 
$\Theta_2=\{0, 2/3\}$. Therefore, in our model, the first agent is relatively selfish with prosocial coefficient or type $1/4$, whereas the second agent can either be completely selfish, with type $0$, or relatively prosocial with type $2/3$. We have 
$	u_1(y,x,1/4,2/3)=3x/4-y/2 $
and $ u_2(y,x,1/4,2/3)=x/3-y/3$. 
Thus, we obtain:
$$\arg\max_{x} u_2(y,x,1/4,2/3)=\{3y\}, \qquad \arg\max_{x} u_2(y,x,1/4,0)=\{0\}.$$
Hence, let
\begin{align*}
s_2(\theta_2)=\left\{
\begin{array}{ccc}
0& & \theta_2=0   \\ 
3y& & \theta_2=2/3
\end{array}
\right.
\end{align*}
We have $u_1(y,0,1/4,\theta_2)=-y/2$ and $u_1(y,3y,1/4,\theta_2)=7y/4. $
Let conditional probability distribution $p(2/3|1/4)=p_0$ and $p(0|1/4)=1-p_0$.
As a result, $$\sum_{\theta_2\in \Theta_2}p(\theta_2|1/4)u_1(y,s_2(\theta_2),1/4,\theta_2)=p_0(-9y/4)+7y/4.$$
Therefore
\begin{align*}
s_1(1/4)=\left\{
\begin{array}{ccc}
1& & p_0<7/9   \\ 
y& & p_0=7/9   \\
0& & p_0>7/9
\end{array}
\right.	
\end{align*}
Hence $(s_1(1/4), s_2(\theta_2))$ is a sub-game perfect equilibrium for the DG. We now see that, depending on its belief about agent 2, agent 1 can send any amount of money to agent 2 and agent 2 can return different amounts of money as an optimal solution for the Trust DG.

\section{Other proofs, results and examples}\label{appendix}
In this section, we first provide the proofs of the two lemmas in the paper and then indicate how multi-games can be employed to model production  in multi-markets.

\subsection{Proofs of the two lemmas} 
\textbf{Lemma\ref{Simlex-Make}}
	If an own-type-linear Bayesian game $G$ is  type-regular on $V \times \{\theta_{-i}\}$ for a given $i\in I$ and $\theta_{-i} \in \Theta_{-i}$, then it is type-regular on  $\Theta_i \times \{\theta_{-i}\}$.
\begin{proof}
	By Theorem~\ref{MG-equivalent}, $G$ is equivalent with a multi-game. Thus, we can assume $G$  is a type-regular multi-game on $V \times \{\theta_{-i}\}$ and therefore  there exists
	a map $\sigma^*_i(\cdot): V\to \Delta(A_i)$  and $\sigma_{-i}  \in \Delta(A_{-i})$ 
	such that for each $1\leq j \leq m$, the strategy profile
	$(\sigma^*_i(v_j), \sigma_{-i} )$ is a NE for the local game $G_{(v_j, \theta_{-i})}$.
	Thus, $\sigma^*_i(v_j)$ is a probability distribution on $A_i$ with $\sum_{a_i \in A_i} \sigma^*_i(v_j)(a_i)=1$.
	We extend the map $\sigma^*_i$ to a map $\overline{\sigma^*_i}(\cdot): \Theta_{i} \to \Delta(A_i)$ by 
	$\overline{\sigma^*_i}(\theta_{i})=	\sum_{j=1}^{m}\theta_{ij}\sigma^*_i(v_j)$ where 
	$\theta_i=(\theta_{i1}, \ldots, \theta_{im})^t$.
	Since
	\begin{align*}
	\sum_{a_i \in A_i}\Big(\sum_{j=1}^{m}\theta_{ij}\sigma^*_i(v_j)\Big)(a_i)& =
	\sum_{a_i \in A_i}\sum_{j=1}^{m}\theta_{ij}\sigma^*_i(v_j)(a_i)
	=\sum_{j=1}^{m}\theta_{ij}\sum_{a_i \in A_i} \sigma^*_i(v_j)(a_i)
	=\sum_{j=1}^{m}\theta_{ij}=1		
	\end{align*}
	the map $\overline{\sigma^*_i}(\cdot)$ is well-defined. We claim that the strategy profile 
	$(\overline{\sigma^*_1}(\theta_{1}), \ldots, \overline{\sigma^*_n}(\theta_{n}))$
	is a NE for
	the local game $G_{(\theta_i, \theta_{-i})}$ for $\theta_i\in \Theta_i$. 
	We have
	\begin{align*}
	u_i\Big(\overline{\sigma^*_i}(\theta_{i}), \sigma_{-i}, \theta_i, \theta_{-i}\Big)
	&=\sum_{j=1}^{m}\theta_{ij}	u_i\big(\sigma^*_i(v_j), \sigma_{-i}, \theta_i, \theta_{-i}\big)
	=\sum_{j=1}^{m}\theta_{ij}\sum_{j=1}^{m}\theta_{ij}
	u_{ij}\big(\sigma^*_i(v_j), \sigma_{-i}\big)\\&
	=\sum_{j=1}^{m}\theta_{ij}\sum_{j=1}^{m}\theta_{ij}	u_{i}\big(\sigma^*_i(v_j), \sigma_{-i},v_j ,\theta_{-i}\big)
	\end{align*}
	Since the strategy profile  $(\sigma^*_i(v_j), \sigma_{-i} )$ is a NE for the local game $G_{(v_j, \theta_{-i})}$,  it follows that
	$$u_{i}\big(\sigma^*_i(v_j), \sigma_{-i},v_j ,\theta_{-i}\big) \geq 
	u_{i}\big(\sigma_i, \sigma_{-i},v_j ,\theta_{-i}\big)
	$$ 	for any given $\sigma_i \in \Delta(A_i)$
	which yields:
	\begin{align*}
	u_i\Big(\overline{\sigma^*_i}(\theta_{i}), \sigma_{-i}, \theta_i, \theta_{-i}\Big)&
	\geq \sum_{j=1}^{m}\theta_{ij}\sum_{j=1}^{m}\theta_{ij}	u_{i}(\sigma_i, \sigma_{-i},v_j ,\theta_{-i})= \sum_{j=1}^{m} \theta_{ij} u_i(\sigma_i, \sigma_{-i},\theta_{i} ,\theta_{-i})\\&
	=u_i\Big(\sum_{j=1}^{m}\theta_{ij}\sigma_i, \sigma_{-i},\theta_{i} ,\theta_{-i}\Big)=u_i(\sigma_i, \sigma_{-i},\theta_{i} ,\theta_{-i})
	\end{align*}
	Hence  $\big(\overline{\sigma^*_i}(\theta_{i}), \sigma_{-i}\big)$ is a NE for 	$G_{(\theta_{i}, \theta_{-i})}$ for each $\theta_i\in \Theta_i$, i.e., $G$ is type-regular on $\Theta_i \times \{\theta_{-i}\}$.
\end{proof}
\noindent
\textbf{Lemma~\ref{Simlex-Make-Booundary}}
	If an own-type-linear Bayesian game  $G$ is type-regular on $V^n$, then $G$ is type-regular on the boundary set  $\bigcup_{i\in I} \Theta_i\times V^{n-1}$.
\begin{proof}
	Suppose $G$ is type-regular on $V^n$ with a map $\sigma^*_i(\cdot): V \to \Delta(A_i)$ for each $i \in I$ such that
	$(\sigma^*_1(v_{j_1}), \ldots, \sigma^*_n(v_{j_n}))$ is a NE for the local game $G_{(v_{j_1}, \ldots, v_{j_n})}$ where $1 \leq j_i \leq m $ and $i\in I$.
	By  Corollary~\ref{extention-regularity}, $\sigma^*_i(\cdot)$ can be extended to 
	$\overline{\sigma^*_i}(\cdot): \Theta_{i} \to \Delta(A_i)$
	such that 
	$$\Big(\sigma^*_1(v_{j_1}), \ldots, \sigma^*_{i-1}(v_{j_{i-1}}),\overline{\sigma^*_i}(\theta_{i}),\sigma^*_{i+1}
	(v_{j_{i+1}}),\ldots, \sigma^*_{n}(v_{j_{n}})\Big)$$
	is a NE for the local game 
	$G_{(v_{j_{1}},\ldots,v_{j_{i-1}}, \theta_{i},v_{j_{i+1}}, \ldots, \ldots,v_{j_{n}})},$ for all $\theta_{i} \in \Theta_{i},$
	$1 \leq j_t \leq m $ and $t \neq i$. 
	Since the map $\overline{\sigma^*_i}(\cdot)$ is independent  of  $\theta_{-i} \in \Theta_{-i}$, it follows that $G$ is type-regular on $\bigcup_{i\in I} \Theta_i\times V^{n-1}$.
\end{proof}
\subsection{Type-regularity in Double Games}
In the following, we discuss on 2-agent DG with symmetric games~\cite{cheng2004notes} to find conditions on type-regularity on extreme types with pure NEs. 
Assume that pure action profiles $(s_1,s_2)$ and $(s'_1,s'_2)$ are NEs for $G_1$ and $G_2$ respectively.
There are 16 possible cases which are can be summarized in four cases where 
$s_1=s'_1$ and $ s_2=s'_2 $, or $s_1=s'_1$ and $s_2 \neq s'_2$, or $s_1\neq s'_1$ and $s_2=s'_2$, or$s_1 \neq s'_1$ and $s_2 \neq s'_2$.
For each case, a few conditions are needed to ensure type-regularity on $\{0,1\}^2$.  The following result, which can be proved by a simple computation, provides the necessary and sufficient conditions for type-regularity on $\{0,1\}^2$. 
\begin{proposition}\label{regulirity-SDG}
	Let $G$  be a DG with symmetric games $G_1$ and $G_2$ and payoff matrices in Table~\ref{tab:adoption technology}. 
	Then $G$ is type-regular on  $\{0,1\}^2$ with witness of type-regularity based on NEs of $G_1$ and $G_2$ if and only if one of the conditions listed in the following table holds.
	\begin{table}[H]
		\centering
		\begin{tabular}{||c| c| c ||} 
			\hline
			NE for $G_1$ & 	NE for $G_2$ & Conditions \\ [0.5ex] 
			\hline\hline
			$(a_1,a_1)$ & $(a_1,a_1)$ & $a>c$ and  $e>g$  \\ 
			\hline
			$(a_1,a_1)$ & $(a_1,a_2)$ & $a>c$, $b>d$, $e=g$ and $f>h$  \\
			\hline
			$(a_1,a_1)$ & $(a_2,a_1)$& $a>c$, $b>d$, $e=g$ and $f>h$  \\
			\hline
			$(a_1,a_1)$ & $(a_2,a_2)$ & $a>c$, $b>d$, $g>e$ and $h>f$  \\ [1ex] 
			\hline
		\end{tabular}
	\end{table}
\end{proposition}
\begin{table} 
	\centering
	\begin{tabular}{|c|c|c|}
		\cline{2-3} 
		\multicolumn{1}{c|}{} & $a_1$ & $a_2$  \\ \hline
		$a_1$ &\footnotesize$(a,a)$&\footnotesize$(b,c)$\\\hline
		$a_1$ &\footnotesize$(c,b)$&\footnotesize$(d,d)$\\\hline
	\end{tabular}\qquad
	\begin{tabular}{|c|c|c|}
		\cline{2-3} 
		\multicolumn{1}{c|}{} & $a_1$ & $a_2$  \\ \hline
		$a_1$& \footnotesize$(e,e)$&\footnotesize$(f, g)$\\\hline
		$a_2$&\footnotesize$(g, f)$&\footnotesize$(h,h)$\\\hline
	\end{tabular}
	\caption{Payoff matrices for  $G_1$ and $G_2$ in Proposition~\ref{regulirity-SDG}.}
	\label{tab:adoption technology}
\end{table}

A similar result can be deduced when $G_1$ and $G_2$ are not symmetric but there will be more conditions. For instance, we have the following.
\begin{proposition}\label{RDG}
	Suppose $G$ is a DG with basic games $G_1$ and $G_2$ whose payoff matrices are depicted in Table~\ref{tab:adoption technology}. The strategy profiles $(s,u)$ and $(t,v)$ are NEs for $G_1$ and $G_2$ which induce type-regularity on $\{0,1\}^2$ if and only if
	$a_1\geq c_1, h_1\geq f_1, b_1\geq d_1, g_1 \geq e_1,a_2\geq b_2, h_2\geq g_2, f_2\geq e_2$ and $ c_2 \geq d_2$ 
\end{proposition}

\begin{table} 
	\centering
	\begin{tabular}{|c|c|c|}
		\cline{2-3} 
		\multicolumn{1}{c|}{} & $u$ & $v$  \\ \hline
		$s$ &\footnotesize$(a_1,a_2)$&\footnotesize$(b_1,b_2)$\\\hline
		$t$ &\footnotesize$(c_1,c_2)$&\footnotesize$(d_1,d_2)$\\\hline
	\end{tabular}\qquad
	\begin{tabular}{|c|c|c|}
		\cline{2-3} 
		\multicolumn{1}{c|}{} & $u$ & $v$  \\ \hline
		$s$& \footnotesize$(e_1,e_2)$&\footnotesize$(f_1, f_2)$\\\hline
		$t$&\footnotesize$(g_1, g_2)$&\footnotesize$(h_1,h_2)$\\\hline
	\end{tabular}
	\caption{Payoff matrices for $G_1$ and $G_2$ in Proposition~\ref{RDG}.}
	\label{tab:adoption technology}
\end{table}
In the follwing example, basic games are considered parametrized coordination games in~\cite[p. 2]{cooper1999coordination}.
\begin{example}\label{CoordinationGame}
	Let $G$ be a DG whose basic games are coordination games with non-negative utilities as depicted in Table~\ref{tab:Ex:cordination}.
	The strategy  profile $(a_1,a_1)$ is a NE for $G_1$ and $G_2$. By Theorem~\ref{Bounadary}, $G$ is type-regular.
	In addition, by Corollary~\ref{extention-regularity}, the constant strategy map profile $(s_1(\cdot), s_2(\cdot))$ is a pure BNE for $G$ where $s_i(\cdot):[0,1]\to \{a_1,a_2\}$ for $i=1,2$ given by 
	$s_i(\theta_i)=a_1$.	
\end{example}
\begin{table}
	\centering
	%\begin{minipage}{.3\linewidth}
	\begin{tabular}{|c|c|c|}
		\cline{2-3} % scpecifies which column should have border in this row
		\multicolumn{1}{c|}{} & $a_1$ & $a_2$  \\ \hline
		$a_1$ &\footnotesize$(x,x)$&\footnotesize$(x,0)$\\\hline
		$a_2$ &\footnotesize$(0,x)$&\footnotesize$(y,y)$\\\hline
	\end{tabular}\qquad
	%\end{minipage}
	%	\hspace{10mm}
	%\begin{minipage}{.3\linewidth}
	\begin{tabular}{|c|c|c|}
		\cline{2-3} % scpecifies which column should have border in this row
		\multicolumn{1}{c|}{} & $a_1$ & $a_2$  \\ \hline
		$a_1$& \footnotesize$(z,z)$&\footnotesize$(z,0)$\\\hline
		$a_2$&\footnotesize$(0,z)$&\footnotesize$(w,w)$\\\hline
	\end{tabular}
	%\end{minipage}
	\caption{Payoff matrices for $G_1$ and $G_2$ of Example~\ref{CoordinationGame}.}
	\label{tab:Ex:cordination}
\end{table}
\subsection{Production in multi-markets}
Consider $n$ multinational companies which compete in multi-markets consisting of, say, $m$ different markets each with its own rate of return. Assume that in each market $j \in J=\{1,\ldots, m\}$ a given product $s_j$ yields the greatest return but due to the design and manufacturing costs each company has to mroduce the same
product in all the $m$ markets, named $M_1, \ldots, M_m$. In this way, we have a multi-game with $A_i=\{s_j:j\in J\}$ for all $i\in I$ where $\theta_{ij}$ is the investment fraction of
company $i$ in market $j$. 
In addition, the total payoff is reduced to the convex combination of the individual payoff for each market weighted by
the rate of investment for that market.
\begin{table} 
	\centering
	\begin{tabular}{|c|c|c|c|}
		\cline{2-4} % specifies which column should have border in this row
		\multicolumn{1}{c|}{} & $s_1$ & $s_2$ & $s_3$ \\ \hline
		$s_1$ & \footnotesize$(3,4)$&\footnotesize$(6,3)$   &\footnotesize$(7,1)           $          \\\hline
		$s_2$ & \footnotesize$(2,5)$&\footnotesize$(3,2)$ &\footnotesize$(5,3)$ \\\hline
		$s_3$ & \footnotesize$(1,3)$&\footnotesize$(0,2)$   &\footnotesize$(3,0) $ \\\hline
	\end{tabular}\qquad
	\begin{tabular}{|c|c|c|c|}
		\cline{2-4} % scpecifies which column should have border in this row
		\multicolumn{1}{c|}{} & $s_1$ & $s_2$ & $s_3$  \\ \hline
		$s_1$ & \footnotesize$(0,4)$&\footnotesize$(0,8)$ &\footnotesize$(1,1)           $          \\\hline
		$s_2$ & \footnotesize$(6,1)$&\footnotesize$(4,5)$ &\footnotesize$(7,3)$ \\\hline
		$s_3$ & \footnotesize$(0,1)$&\footnotesize$(1,6)$ &\footnotesize$(1,3) $ \\\hline
	\end{tabular} 
	\qquad
	\begin{tabular}{|c|c|c|c|}
		\cline{2-4} % scpecifies which column should have border in this row
		\multicolumn{1}{c|}{} & $s_1$ & $s_2$ & $s_3$  \\ \hline
		$s_1$ & \footnotesize$(1,0)$&\footnotesize$(1,2)$ &\footnotesize$(4,5)           $          \\\hline
		$s_2$ & \footnotesize$(0,1)$&\footnotesize$(3,2)$ &\footnotesize$(3,4)$ \\\hline
		$s_3$ & \footnotesize$(2,4)$&\footnotesize$(5,3)$ &\footnotesize$(6,7) $ \\\hline
	\end{tabular} 
	\caption{Payoff matrices for markets $M_1$, $M_2$ and  $M_3$. }
	\label{tab:Multinational Companies}
\end{table}
We give a numerical example. 
We model the competition of two firms in three markets with  multi-game  such that payoff matrices for agents in each market are shown in Table~\ref{tab:Multinational Companies} with $\Theta_i=\Sigma^2$ for each firm $i=1,2$.
It is easy to check that 
$\sigma^*_i(\cdot): V \to \Delta(A_i)$  given by $\sigma^*_1(v_j)=s_j$ for $j=1,2,3$
is a witness of type-regularity $G$  on $V$. By Corollary~\ref{suffucient and nassacery condition for multi-game}, $G$ is type-regular with type-regularity witness
$\overline{\sigma^*_i}(\cdot): \Theta_{i} \to \Delta(A_i)$ with
$\overline{\sigma^*_i}(\theta_{i})=	\sum_{j=1}^{m}\theta_{ij}\sigma^*_i(v_j).$  Theorem~\ref{motivation-Regularity} implies $(\overline{\sigma^*_1}(\cdot), \overline{\sigma^*_1}(\cdot))$ is a BNE for all prior.

\section{Conclusion}
We have developed the notion of type-regularity for Bayesian games which represents a necessary and sufficient condition for a Bayesian game to have a BNE independent of all priors. We have then shown that an $m$-dimensional $n$-agent Bayesian game in which each agent's utility is linearly dependent on the agent's type, is equivalent to a multi-game, i.e., a simultaneous competition by the agents in $m$ basic $n$-agent games. This result is also extended to generalised multi-games that are equivalent to Bayesian games, in which each agent's utility depends linearly on all agents' types. We have then proven that an own-type-linear Bayesian game is type-regular if and only if its equivalent multi-game is type-regular on the vertices of the $(m-1)$-dimensional simplex, a result which is used in different contexts to construct prior independent BNE.

We can employ multi-games to model decision making by agents investing with their individual weights in multiple environments or markets that are considered as basic games. Multi-games are also proposed as a way to model human rational-social decision making. In particular, we have  constructed a type-regular DG for the PD and a prior dependent perfect subgame NE for a DG based on the Trust game to account for the prosocial component in human decision making.

Here are a number of challenges and questions for further work: (i) type-regularity for type-linear Bayesian games, (ii) type-regularity for Bayesian games with piece-wise linear utilities, (iii) multi-games based on basic games with incomplete information, (iv) Bayesian network games, (v) existence and construction of regular pure BNE, (vi) further applications of linear-type Bayesian games.

%\bibliography{ms}

%----------------------------------------------------------------------------------------

\end{document}